\documentclass[letterpaper,11pt]{article}

\usepackage{times}

\usepackage{amsthm,amsmath,amsfonts}

\usepackage{graphicx}
\usepackage[small]{caption}
\usepackage[hidelinks]{hyperref}
\newcommand{\orcid}[1]{\,\href{https://orcid.org/#1}{\includegraphics[width=8pt]{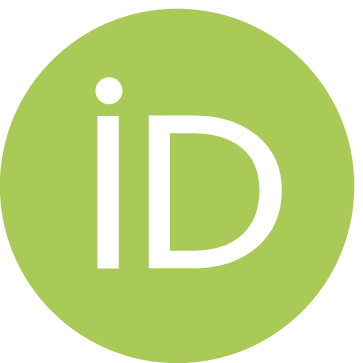}}}

\usepackage{xcolor}
\newcommand{\hl}[1]{\textcolor{red}{#1}}

\renewcommand{\P}{\mathcal{P}}

\newcommand{\ZZ}{\mathbb{Z}}

\newcommand{\kcd}{\textsc{$(k,c)$-Decomposition}}
\newcommand{\mcd}{\textsc{Minimum-Color Decomposition}}
\newcommand{\msd}{\textsc{Minimum-Size Decomposition}}
\newcommand{\kcfc}{\textsc{$(k,c)$-Fragmented Coloring}}
\newcommand{\bp}{\textsc{Bin Packing}}
\newcommand{\ubp}{\textsc{Unary Bin Packing}}

\newtheorem{theorem}{Theorem}
\newtheorem{corollary}{Corollary}
\newtheorem{lemma}{Lemma}
\newtheorem{definition}{Definition}

\title{Decomposing a graph into subgraphs with small components}

\author{Rain Jiang\orcid{0000-0002-0144-942X}\qquad
Kai Jiang\orcid{0000-0001-8165-0571}\qquad
Minghui Jiang\orcid{0000-0003-1843-9292}\,\thanks{\texttt{ dr.minghui.jiang at gmail.com}}\medskip\\
Home School, USA}

\date{}

\begin{document}

\maketitle

\begin{abstract}
The component size of a graph
is the maximum number of edges in any connected component of the graph.
Given a graph $G$ and two integers $k$ and $c$,
$(k,c)$-Decomposition
is the problem of deciding whether
$G$ admits an edge partition into $k$ subgraphs
with component size at most $c$.
We prove that for any fixed $k \ge 2$ and $c \ge 2$,
$(k,c)$-Decomposition
is NP-complete in bipartite graphs.
Also,
when both $k$ and $c$ are part of the input,
$(k,c)$-Decomposition
is NP-complete even in trees.
Moreover,
$(k,c)$-Decomposition
in trees is W[1]-hard with parameter $k$,
and is FPT with parameter $c$.
In addition,
we present approximation algorithms
for decomposing a tree
either into the minimum number of subgraphs with component size at most $c$,
or into $k$ subgraphs minimizing the maximum component size.
En route to these results, we also obtain
a fixed-parameter algorithm for
Bin Packing
with the bin capacity as parameter.
\end{abstract}

\section{Introduction}

Edge coloring is the problem of assigning a color to each edge of a graph such
that no two adjacent edges have the same color.
Any edge coloring of a graph with maximum degree $\Delta$
requires at least $\Delta$ colors.
Vizing's theorem states that
$\Delta + 1$ colors suffice.
The minimum number of colors in any edge coloring of a graph is its
\emph{chromatic index}.
Deciding whether the chromatic index of a graph with maximum degree $\Delta$
is $\Delta$ or $\Delta+1$
is generally a hard computational problem,
although a constructive proof of Vizing's theorem~\cite{MG92}
implies a polynomial-time approximation algorithm with additive error at most $1$.
Holyer~\cite{Ho81} proved that edge coloring cubic graphs
with three colors,
where $\Delta=3$,
is already NP-complete.
Extending this result,
Leven and Galil~\cite{LG83} proved that edge coloring $\Delta$-regular graphs
with $\Delta$ colors is NP-complete for any $\Delta \ge 3$.
Cai and Ellis~\cite{CE91} further proved that
edge coloring $\Delta$-regular line graphs of bipartite graphs
with $\Delta$ colors is NP-complete for any odd $\Delta \ge 3$.
On the other hand,
it is known that
$\Delta$ colors always suffice for edge coloring bipartite graphs;
see~\cite[Proposition 18.1.3]{Ca94} for a simple proof.
Moreover,
edge coloring bipartite graphs (or bipartite multigraphs)
admits an efficient polynomial-time exact algorithm,
even when $\Delta$ is part of the input~\cite{COS01}.

Define the \emph{component size} of a graph as the maximum number of edges
in any connected component of the graph.
In this paper, we study the	following problem that generalizes edge coloring:

\begin{definition}
Given a graph $G$ and two integers $k \ge 2$ and $c \ge 1$,
\kcd\ is the problem of deciding whether
$G$ admits an edge partition into $k$ subgraphs
with component size at most $c$.
\end{definition}

We show that graph decomposition becomes hard, even in bipartite graphs,
when the upper bound on the component size of the resulting subgraphs is
relaxed from $1$ (as in edge coloring) to any fixed $c \ge 2$:

\begin{theorem}\label{thm:k2c2}
For any fixed $k \ge 2$ and $c \ge 2$,
\kcd\ is NP-complete in bipartite graphs.
\end{theorem}

The \emph{arboricity} of a graph is the minimum number of parts
in an edge partition of the graph such that each part induces a
forest~\cite{Na64}.
The trees in the forests may be restricted to various subclasses of trees,
leading to several related concepts of arboricity.
For example,
the
\emph{linear arboricity}
(respectively, \emph{star arboricity})
of a graph is the minimum number of parts
in an edge partition of the graph such that each part induces
a disjoint union of paths (respectively, stars).
While the arboricity of any graph can be computed
in polynomial time~\cite{GW92},
the arboricities with various restrictions on the trees
often turn out to be hard to compute~\cite{Ji18}.
Our proof of Theorem~\ref{thm:k2c2} shows that
arboricity and star arboricity with bounded component size are also hard to compute:

\begin{corollary}\label{cor:arboricity}
For any fixed $k \ge 2$ and $c \ge 2$,
deciding whether a bipartite graph
admits an edge partition into $k$ forests with component size at most $c$
is NP-complete.
\end{corollary}

\begin{corollary}\label{cor:stararboricity}
For any fixed $k \ge 2$ and $c \ge 2$,
deciding whether a bipartite graph
admits an edge partition into $k$ forests of stars with component size at most $c$
is NP-complete.
\end{corollary}

The \emph{linear $c$-arboricity} of a graph is the minimum number of parts
in an edge partition of the graph into \emph{linear $c$-forests}, i.e.,
disjoint unions of paths of length at most $c$~\cite{BFHP84}.
Note that linear $1$-arboricity is simply chromatic index.
Thus the aforementioned results on edge coloring~\cite{Ho81,LG83,CE91}
imply that determining the linear $1$-arboricity of a graph is hard.
Bermond, Fouquet, Habib, and Peroche~\cite{BFHP84} proved that
deciding whether a cubic graph has linear $3$-arboricity at most $2$
is NP-complete, and conjectured that determining the linear $c$-arboricity
of a graph is hard for all $c$.
Since linear $2$-arboricity is equivalent to
arboricity and star arboricity with component size at most $2$,
the $c = 2$ case of this conjecture is confirmed by
our Corollary~\ref{cor:arboricity}
and Corollary~\ref{cor:stararboricity}:

\begin{corollary}\label{cor:linear2}
For any fixed $k \ge 2$,
deciding whether a bipartite graph
has linear $2$-arboricity $k$
is NP-complete.
\end{corollary}

For two graphs $G$ and $H$,
an \emph{$H$-decomposition} of $G$ is an edge partition of $G$
into subgraphs isomorphic to $H$.
For any fixed graph $H$,
the problem of deciding whether an input graph $G$ admits an $H$-decomposition
is NP-complete when $H$ has a component of at least three edges~\cite{DT97},
and is polynomially solvable when every component of $H$ has at most two
edges~\cite{BL09}, in contrast to Corollary~\ref{cor:linear2}.

\bigskip
We established in Theorem~\ref{thm:k2c2} that
\kcd\ is NP-complete in bipartite graphs.
It is natural to ask whether the problem remains hard in simpler graphs.
Our next theorem characterizes the complexity of \kcd\ in trees:

\begin{theorem}\label{thm:tree}
When both $k$ and $c$ are part of the input,
\kcd\ in trees is NP-complete,
is W[1]-hard with parameter $k$,
and is FPT with parameter $c$.
\end{theorem}

Theorem~\ref{thm:tree} implies that
for any fixed $c$, \kcd\ in trees can be solved in polynomial time.
For the related arboricity problem of deciding whether a tree admits
an edge partition into $k$ linear $c$-forests,
Chang, Chen, Fu, and Huang~\cite{CCFH00} presented an algorithm that runs in
polynomial time when $c$ is fixed.

\bigskip
The decision problem \kcd\ may be extended to two different optimization problems:

\begin{definition}
Given a graph $G$ and an integer $c \ge 1$, \mcd\
is the problem of decomposing $G$ into
the minimum number of subgraphs with component size at most $c$.
\end{definition}

\begin{definition}
Given a graph $G$ and an integer $k \ge 2$, \msd\
is the problem of decomposing $G$ into
$k$ subgraphs minimizing the maximum component size.
\end{definition}

Theorem~\ref{thm:k2c2} implies that \mcd\ for any fixed $c \ge 2$
and \msd\ for any fixed $k \ge 2$ are APX-hard
and inapproximable within a factor of $3/2$
in bipartite graphs.
The previous results on edge coloring~\cite{Ho81,LG83,CE91}
imply that \mcd\ for $c = 1$ is APX-hard
and inapproximable within a factor of $4/3$ in general graphs.
In the next two theorems,
we show that both minimization problems can be approximated very well in
trees:

\begin{theorem}\label{thm:mcd}
There is a linear-time approximation algorithm for \mcd\ in trees
that finds an edge
partition of any tree $T$ into at most $k^* + 1$ subgraphs
with component size at most $c$,
where $k^*$ is the smallest number such that
$T$ can be decomposed into $k^*$ subgraphs with component size at most $c$.
\end{theorem}

\begin{theorem}\label{thm:msd}
There is a polynomial-time approximation scheme for \msd\ in trees
that finds an edge partition of any tree $T$ into $k$ subgraphs with
component size at most $(1+\epsilon)c^*$,
where $c^*$ is the smallest number such that
$T$ can be decomposed into $k$ subgraphs with component size at most $c^*$.
\end{theorem}

Kleinberg, Motwani, Raghavan, and Venkatasubramanian~\cite{KMRV97}
introduced
\kcfc,
the problem of vertex partitioning a graph into $k$ induced subgraphs
such that the number of vertices in each component of each subgraph
is at most $c$.
Our problem \kcd\ is equivalent to \kcfc\ in line graphs,
since edge partitioning a graph $G$ is the same as
vertex partitioning the line graph $L(G)$.

The problem \kcfc\ in general graphs
is easily shown to be NP-hard for any $k \ge 2$ and $c \ge 2$,
because the graph property of at most $c$ vertices in each component
is additive and induced-hereditary,
while vertex partitioning into fixed additive induced-hereditary properties is
NP-hard in general~\cite{Fa04},
except the special case of bipartite testing,
i.e.,
vertex partitioning into two independent sets.
Also, the problem is trivial in bipartite graphs (and hence in trees),
since every bipartite graph has a vertex partition into two induced subgraphs in which
every component has only one vertex.

\bigskip
Our proofs of Theorem~\ref{thm:tree},
Theorem~\ref{thm:mcd}, and Theorem~\ref{thm:msd}
are based on a close relationship
between \kcd\ in trees and the classical problem of \bp:

\begin{definition}\label{def:bp}
Given $n$ items of integer weights $w_i$, $1 \le i \le n$,
and $k$ bins of integer capacity $c$,
\bp\ is the problem of deciding
whether the $n$ items can be packed into the $k$ bins,
that is, whether the set of $n$ items
can be partitioned into $k$ subsets,
such that the total weight of the items in each subset is at most $c$.
\ubp\ is the version of \bp\ where all integers in the problem instance are encoded in unary.
\end{definition}

It is well-known that \bp\ is strongly NP-hard~\cite{GJ79},
and hence \ubp\ is NP-hard,
when both $k$ and $c$ are part of the input.
In terms of parameterized complexity,
Jansen et~al.~\cite{JKMS13} proved that \ubp,
even with the condition $\sum_{i=1}^n w_i = k c$,
is already W[1]-hard with parameter $k$.
In the following theorem,
we obtain a fixed-parameter algorithm for \bp\ with parameter $c$:

\begin{theorem}\label{thm:fpt}
\bp\ admits an exact algorithm running in
$2^{O(c^{3/2})}N^{O(1)}$ time,
where $N$ is the size of the problem instance encoded in binary.
\end{theorem}

Hochbaum and Shmoys~\cite{HS87}
introduced the notion of \emph{dual approximation algorithms}
that approximate the feasibility of a problem rather than its optimality.
The existence of dual approximation algorithms for \bp,
similar to the one in the following theorem,
are known~\cite[Sections 9.3.2.1 and 9.3.2.2]{Ho97}.
We obtain a simple alternative via Theorem~\ref{thm:fpt}.

\begin{theorem}\label{thm:dual}
\bp\ admits a dual approximation algorithm
that, given $n$ items and $k$ bins of capacity $c$,
and given $0 < \epsilon \le 1$,
either
decides that the $n$ items can be packed into $k$ enlarged bins
of capacity $(1+\epsilon)c$ and returns yes,
or
decides that the $n$ items cannot be packed into $k$ bins of capacity $c$
and returns no,
in
$2^{O(\epsilon^{-3})}N^{O(1)}$ time,
where $N$ is the size of the problem instance encoded in binary.
\end{theorem}

Throughout the paper, the \emph{size} of a graph is the number of edges in it.
We may refer to an edge partition of a graph into $k$ subgraphs
and the corresponding $k$-color assignment to the edges interchangeably.

\section{Hardness of decomposing bipartite graphs}

In this section we prove Theorem~\ref{thm:k2c2},
Corollary~\ref{cor:arboricity},
and
Corollary~\ref{cor:stararboricity}.
Fix any $k \ge 2$ and $c \ge 2$.
The problem \kcd\ and the related arboricity problems are clearly in NP\@.
We prove their NP-hardness in bipartite graphs in the following.

\subsection{Building blocks $H_i$}

The building blocks of our construction are a series of bipartite graphs
$H_i = H_i(k,c)$ for $i \ge 0$.
$H_0$ is the star $K_{1,kc}$ with one edge designated as an \emph{outlet}.
For $i > 0$, $H_i$ has $c^i$ outlets and is constructed recursively as follows.
Take $k-1$ disjoint copies of $H_{i-1}$
and $c^{i-1}$ disjoint copies of $K_{1,c}$.
Attach a distinct outlet from each copy of
$H_{i-1}$
to the center of each copy of $K_{1,c}$.
Then designate the edges of all copies of $K_{1,c}$ in the resulting graph
$H_i$
as its outlets.
Refer to Figures~\ref{fig:k2c3} and~\ref{fig:k3c2} for some examples.

It is easy to verify that $H_i$ is bipartite.
Since it includes $(k-1)^i$ copies of $K_{1,kc}$
and $\sum_{j = 1}^i c^{j-1}(k-1)^{i - j}$ copies of $K_{1,c}$,
its size is at most
\begin{equation}\label{eq:H}
(k-1)^i kc + \sum_{j = 1}^i c^j(k-1)^{i - j}
\le
(k-1)^i kc + (c+k)^i
\le
(c+k)^i(kc + 1)
\le
(c+k)^{i+2}.
\end{equation}

\begin{figure}[htbp]
\centering\includegraphics{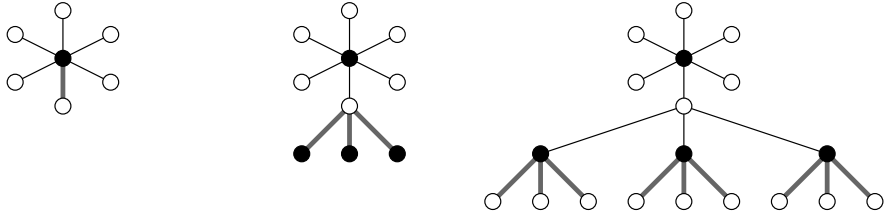}
\caption{$H_0$, $H_1$, and $H_2$ for $k = 2$ and $c = 3$.}
\label{fig:k2c3}
\end{figure}

\begin{figure}[htbp]
\centering\includegraphics{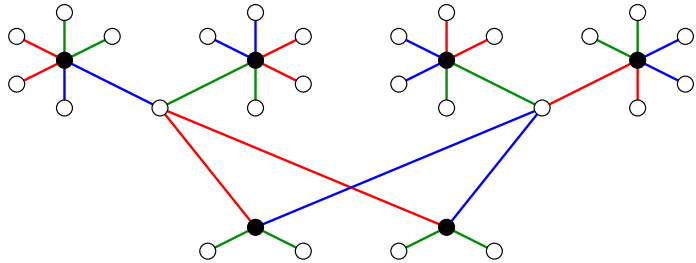}
\caption{$H_2$ for $k = 3$ and $c = 2$.}
\label{fig:k3c2}
\end{figure}

\begin{lemma}\label{lem:H}
$H_i$ admits an edge partition into $k$ subgraphs
with component size at most $c$.
Moreover,
in any edge partition of $H_i$ into $k$ subgraphs
with component size at most $c$,
each component must be a star $K_{1,c}$,
and the $c^i$ outlets must belong to the same subgraph.
\end{lemma}

\begin{proof}
We prove the lemma by induction on $i$.
For the base case when $i = 0$,
the lemma clearly holds for $H_0 = K_{1,kc}$.
Now let $i \ge 1$ for the inductive step.

We first show that $H_i$ admits an edge partition into $k$ subgraphs
with component size at most $c$.
By the induction hypothesis,
$H_{i-1}$ admits an edge partition into $k$ subgraphs
in which each component is a star $K_{1,c}$,
and the $c^{i-1}$ outlets are in the same subgraph.
Consider the $k$-color assignment to the edges of $H_{i-1}$
corresponding to this edge partition.
Then by rotation of colors,
there is an edge partition of the union of the $k-1$ copies of $H_{i-1}$ in
$H_i$,
such that the outlets from different copies have different colors.
Thus the outlets from the $k-1$ copies of $H_{i-1}$
have $k-1$ different colors.
Assign the remaining color to all edges of the $c^{i-1}$ disjoint copies of
$K_{1,c}$ in $H_i$.
This color assignment corresponds to an edge partition of $H_i$ into $k$ subgraphs
in which each component is a star $K_{1,c}$,
and the $c^i$ outlets are all in the same subgraph.

We next show that
in any edge partition of $H_i$ into $k$ subgraphs
with component size at most $c$,
each component must be a star $K_{1,c}$,
and the $c^i$ outlets must be in the same subgraph.
Consider an arbitrary edge partition of $H_i$ into $k$ subgraphs with at
most $c$ edges in each component,
and the corresponding $k$-color assignment to the edges of $H_i$.
Then in the derived edge partition of each copy of $H_{i-1}$ in $H_i$,
each component is a star $K_{1,c}$,
and all $c^{i-1}$ outlets are in the same subgraph,
by the induction hypothesis.
Moreover,
since each component in each copy of $H_{i-1}$
already has the maximum allowable size $c$,
the outlets from different copies of $H_{i-1}$ must have different colors
to avoid forming larger components.
Then all edges of the $c^{i-1}$ disjoint copies of $K_{1,c}$ in $H_i$
must have the only remaining color.
\end{proof}

\subsection{Reduction for $k = 2$}

Our proof for the $k = 2$ case of Theorem~\ref{thm:k2c2},
Corollary~\ref{cor:arboricity},
and
Corollary~\ref{cor:stararboricity}
is based on a reduction from the NP-complete problem
\textsc{2-Colorability of 3-Uniform Hypergraphs}~\cite{Lo73}.

A \emph{$q$-uniform hypergraph} consists of a set $V$ of vertices
and a set $E$ of hyperedges,
where each hyperedge in $E$ is a subset of exactly $q$ vertices in $V$.
A hypergraph is \emph{$k$-colorable}
if its vertices can be colored with $k$ colors such that in each hyperedge,
not all $q$ vertices have the same color.

Given a $3$-uniform hypergraph $G$,
we will construct a bipartite graph $G_2$ with maximum degree $2c$ such that
$G$ is $2$-colorable
if and only if
$G_2$ admits an edge partition into $2$ subgraphs
with component size at most $c$.

Let $\Delta$ be the maximum degree of $G$, that is,
the maximum number of hyperedges in which a vertex is contained.
Let $j$ be the smallest integer such that $c^j \ge (c-1)\Delta$.

Include in $G_2$ a distinct copy of $H_j$ for each vertex in $G$,
and a distinct copy of $K_{1,c+1}$ for each hyperedge in $G$.
For each hyperedge in $G$, which consists of three vertices,
take $c+1$ distinct outlets
from the corresponding three copies of $H_j$,
with at least one outlet from each copy,
then attach the $c+1$ outlets to the $c+1$ leaves
of the copy of $K_{1,c+1}$ corresponding to the hyperedge,
with one outlet to each leaf.
This completes the construction of $G_2$.

Recall that $H_j$ has $c^j$ outlets.
Since $c^j \ge (c-1)\Delta$,
each copy of $H_j$
can contribute at least $c-1$ distinct outlets to each adjacent copy of $K_{1,c+1}$.
On the other hand,
of the three copies of $H_j$ adjacent to each copy of $K_{1,c+1}$,
each must contribute at least $1$
and hence at most $c+1 - 2 = c-1$ of the $c+1$ outlets.
Thus there are enough outlets from each copy of $H_j$
to the adjacent copies of $K_{1,c+1}$.

By our choice of $j$,
we have
$c^{j-1} < (c-1)\Delta$ and hence $c^j < c^2 \Delta$.
Recall~\eqref{eq:H} that the size of $H_j$ is at most $(c+k)^{j+2}$.
Since $k = 2$ and $c \ge 2$, we have
$$
(c+k)^{j+2} \le (c^2)^{j+2} = c^4 (c^j)^2 < c^4 (c^2 \Delta)^2 = c^8 \Delta^2.
$$
The reduction is clearly polynomial.
The following lemma completes the proof for the $k = 2$ case of
Theorem~\ref{thm:k2c2}, Corollary~\ref{cor:arboricity}
and Corollary~\ref{cor:stararboricity}:

\begin{lemma}\label{lem:k2}
$G$ is $2$-colorable
if and only if
$G_2$ admits an edge partition into two subgraphs
with component size at most $c$,
if and only if
$G_2$ admits an edge partition into two forests
with component size at most $c$,
if and only if
$G_2$ admits an edge partition into two forests of stars
with component size at most $c$.
\end{lemma}

\begin{proof}
It suffices to prove a cycle of four implications:
\begin{enumerate}\setlength\itemsep{0pt}
\item
if $G$ is $2$-colorable,
then $G_2$ admits an edge partition into two forests of stars
with component size at most $c$,
\item
if $G_2$ admits an edge partition into two forests of stars
with component size at most $c$,
then $G_2$ admits an edge partition into two forests
with component size at most $c$,
\item
if $G_2$ admits an edge partition into two forests
with component size at most $c$,
then $G_2$ admits an edge partition into two subgraphs
with component size at most $c$,
\item
if $G_2$ admits an edge partition into two subgraphs
with component size at most $c$,
then $G$ is $2$-colorable.
\end{enumerate}

We first prove implication 1.
Suppose that there is a $2$-coloring of the vertices of $G$ such that
in each hyperedge,
not all three vertices have the same color.
Color the edges in each copy of $H_j$ in $G_2$ with two colors
to partition them into two subgraphs as in Lemma~\ref{lem:H},
such that all $c^j$ outlets in each copy of $H_j$
have the same color as the corresponding vertex in $G$.
Next color the edges in each copy of $K_{1,c+1}$ in $G_2$ corresponding to a hyperedge in $G$,
such that each edge of $K_{1,c+1}$ has a color different
(note that there are only two colors available)
from the adjacent outlet.
Since in each hyperedge of $G$ not all three vertices have the same color,
and since the $c+1$ outlets adjacent to the $c+1$ edges of $K_{1,c+1}$ include
at least one outlet from each of the three copies of $H_j$,
it follows that not all $c+1$ edges have the same color.
Thus we obtain an edge partition of $G_2$ into two subgraphs
with component size at most $c$.
Note that the two subgraphs are indeed two forests of stars.

Implications 2 and 3 are trivial.

We next prove implication 4.
Suppose that $G_2$ admits an edge partition into two subgraphs
with component size at most $c$.
Color the edges of $G_2$ with two colors according to this partition.
Then it follows by Lemma~\ref{lem:H} that
in each copy of $H_j$ in $G_2$,
each component is a star with $c$ edges,
and all outlets have the same color.
Assign this color to the corresponding vertex in $G$.
Since each outlet of $H_j$ is already in a component of the maximum allowable size $c$,
each edge of $K_{1,c+1}$ must have a color different from its adjacent outlet.
In each copy of $K_{1,c+1}$,
not all $c+1$ edges can have the same color
under the constraint of component size at most $c$.
Correspondingly,
in each hyperedge of $G$,
not all three vertices have the same color either.
Thus $G$ is $2$-colorable.
\end{proof}

\subsection{Reduction for $k \ge 3$}

Our proof for the $k \ge 3$ case of Theorem~\ref{thm:k2c2},
Corollary~\ref{cor:arboricity},
and
Corollary~\ref{cor:stararboricity}
is based on a reduction from the NP-complete problem
\textsc{$k$-Colorability},
which asks whether a given graph $G$ admits a vertex coloring with $k$ colors.
Maffray and Preissmann~\cite{MP96} proved that for any $k \ge 3$,
\textsc{$k$-Colorability} is NP-hard even in triangle-free graphs.
Emden-Weinert, Hougardy, and Kreuter~\cite{EHK98}
proved that for any $k \ge 3$,
\textsc{$k$-Colorability} is NP-hard in graphs with maximum degree
$k + \lceil\sqrt{k}\,\rceil - 1$.
Note that \textsc{$k$-Colorability} for $k = 2$ is just bipartite testing
which is well-known to be solvable in polynomial time.
Thus we had to prove the $k = 2$ case by reduction from a different problem.

Let $G$ be an input graph for \textsc{$k$-Colorability}.
Construct $G_k$ as follows.
Let $j$ be the smallest integer such that $c^j$ is at least
the maximum degree of $G$.
For each vertex in $G$,
include in $G_k$ a distinct copy of $H_j$.
For each edge in $G$, which consists of two vertices,
take a distinct outlet from the copy of $H_j$ for each vertex,
then join the two ends of the two outlets into one vertex.

It is easy to verify that $G_k$ is bipartite, and the reduction is polynomial.
Moreover,
by a similar argument as in the proof of Lemma~\ref{lem:k2},
$G$ is $k$-colorable
if and only if
$G_k$ admits an edge partition into $k$ subgraphs
with component size at most $c$,
if and only if
$G_k$ admits an edge partition into $k$ forests
with component size at most $c$,
if and only if
$G_k$ admits an edge partition into $k$ forests of stars
with component size at most $c$.

\bigskip
This completes the proof of Theorem~\ref{thm:k2c2},
Corollary~\ref{cor:arboricity},
and Corollary~\ref{cor:stararboricity}.

\section{Fixed-parameter algorithm for bin packing}

In this section we prove Theorem~\ref{thm:fpt} by
obtaining a fixed-parameter algorithm for \bp\ with parameter $c$.

Recall Definition~\ref{def:bp} that the input to \bp\ consists of
$n$ items of integer weights $w_i$, $1 \le i \le n$,
and $k$ bins of integer capacity $c$.
The problem has a solution only if $\max\{ w_i \mid 1 \le i \le n\} \le c$ and
$\sum_{i=1}^n w_i \le k c$,
which we assume.
Let $N$ be the size of the problem instance encoded in binary.
Then
$N = \Theta(n + \sum_{i=1}^n \log w_i + \log k + \log c)$.
In particular, we have $\log kc = O(N)$.

For $1 \le w \le c$, let $a_w$ be the number of items with weight $w$.
In $N^{O(1)}$ time, we can compute $a_w$ for all $w$.
Then $\sum_{w=1}^c w\,a_w = \sum_{i=1}^n w_i \le k c$.
With the multiplicities $a_w$ for $1 \le w \le c$,
we no longer need $w_i$ for $1 \le i \le n$.
Henceforth, we work on the $c$ integers $a_w$ only.
Without loss of generality,
we increase $a_1$
until $\sum_{w=1}^c w\,a_w = k c$.
Note that each $a_w$ has magnitude at most $kc$, where $\log kc = O(N)$.
It remains to decide whether the items with multiplicities $a_w$ can be partitioned
into $k$ subsets, each of total weight exactly $c$.

Consider the two vectors
$\boldsymbol{w} = (1,2,\ldots,c)$
and
$\boldsymbol{a} = (a_1,a_2,\ldots,a_c)$
in $\ZZ_{\ge0}^c$.
Then the condition $\sum_{w=1}^c w\,a_w = kc$
becomes $\boldsymbol{w}\cdot\boldsymbol{a} = kc$.
For $b \ge 0$,
let
$$
\P_b = \{\,
\boldsymbol{x} \in \ZZ_{\ge 0}^c
\mid \boldsymbol{w}\cdot \boldsymbol{x} = b
\,\}.
$$
Then $\P_c$ is the set of vectors corresponding to
all feasible \emph{patterns} of item multiplicities for a single bin,
and we can reformulate the problem as deciding whether there exists
a multiset of vectors in $\P_c$
whose sum is $\boldsymbol{a}$.

Recall that the \emph{partition number} $p(n)$ for $n \ge 0$
is the number of multisets of positive integers summing up to $n$.
In particular, $p(0) = 1$.
Then, by definition, the size of $\P_c$ is exactly $p(c)$.
Moreover, the size of $\P_b$ is at most $p(b)$ for all $b \ge 0$.

It is known that $p(n) = 2^{O(\sqrt{n})}$~\cite[Theorem 15.7]{LW01}.
Let $d(c) \ge 0$ be the smallest integer such that
$2^d > \sum_{h=0}^d p(hc)$ for all $d > d(c)$.
Since $\sum_{h=0}^d p(hc) = 2^{O(\sqrt{d c})}$,
we have $d(c) = O(c)$.

Let $\boldsymbol{v}_i$, $1 \le i \le p(c)$, be the vectors in $\P_c$.
For any index set $I \subseteq \{ 1,\ldots, p(c) \}$,
the vector $\sum_{i\in I} \lambda_i \boldsymbol{v}_i$,
where $\lambda_i \in \ZZ_{\ge 0}$ for all $i \in I$,
is called an \emph{integer conical combination} of vectors in $\P_c$.
Consider the \emph{integer conical hull} of $\P_c$
consisting of all integer conical combinations of vectors in $\P_c$:
$$
\mathrm{int.cone}(\P_c) = \bigg\{\, \sum_{i=1}^{p(c)} \lambda_i \boldsymbol{v}_i
\;\Big\vert\;
\lambda_i \in \ZZ_{\ge 0} \textrm{ for } 1 \le i \le p(c)
\bigg\}.
$$

We prove the following lemma using a common technique for Carath\'eodory bounds;
see for example~\cite[Lemma 3]{ES06}.

\begin{lemma}\label{lem:cone}
Every vector in
$\mathrm{int.cone}(\P_c)$
can be represented as the integer conical combination of
at most $d(c)$ distinct vectors in $\P_c$.
\end{lemma}

\begin{proof}
Let $I$ be any set of $d$ distinct indices from $\{ 1,\ldots, p(c) \}$,
where $0 \le d \le p(c)$.
Then $I$ has $2^d$ distinct subsets.
Denote by $|H|$ the size of a set $H$.
For each subset $H \subseteq I$, where $0 \le |H| \le d$,
the vector
$\boldsymbol{x} = \sum_{i\in H} \boldsymbol{v}_i$
satisfies
$$
\boldsymbol{w}\cdot\boldsymbol{x}
= \boldsymbol{w}\cdot \sum_{i\in H} \boldsymbol{v}_i
= \sum_{i\in H}
\boldsymbol{w}\cdot \boldsymbol{v}_i
= \sum_{i\in H} c
= |H|\,c,
$$
and hence
$\boldsymbol{x} \in \P_{hc}$
with $h = |H|$.
Thus $\boldsymbol{x}$ is one of the vectors in
$\cup_{h=0}^d \P_{hc}$.

Consider any vector
$\boldsymbol{u} = \sum_{i\in I} \lambda_i \boldsymbol{v}_i$,
where $\lambda_i \in \ZZ_{> 0}$ for all $i\in I$.
Suppose that $d > d(c)$.
Then by our choice of $d(c)$, 
we have
$2^d > \sum_{h=0}^d p(hc)$.
Note that the number of vectors in
$\cup_{h=0}^d \P_{hc}$
is at most $\sum_{h=0}^d p(hc)$.
Thus the number of distinct subsets of $I$
exceeds the number of vectors in $\cup_{h=0}^d \P_{hc}$.
By the pigeonhole principle,
there are two distinct subsets $A$ and $B$ of $I$
such that
$\sum_{i\in A} \boldsymbol{v}_i = \sum_{i\in B} \boldsymbol{v}_i$.
Then $A' = A\setminus B$ and $B' = B\setminus A$
are two distinct and disjoint subsets of $I$
such that
$\sum_{i\in A'} \boldsymbol{v}_i = \sum_{i\in B'} \boldsymbol{v}_i$.
Assume without loss of generality that $A' \neq \emptyset$.

Let $\lambda = \min\{\, \lambda_i \mid i \in A' \,\}$.
We can rewrite $\boldsymbol{u}$ as follows:
\begin{align*}
\sum_{i\in I} \lambda_i \boldsymbol{v}_i
&= \sum_{i\in I\setminus(A'\cup B')} \lambda_i \boldsymbol{v}_i
+ \sum_{i\in A'} \lambda_i \boldsymbol{v}_i
+ \sum_{i\in B'} \lambda_i \boldsymbol{v}_i\\
&= \sum_{i\in I\setminus(A'\cup B')} \lambda_i \boldsymbol{v}_i
+ \sum_{i\in A'} (\lambda_i - \lambda) \boldsymbol{v}_i
+ \sum_{i\in B'} (\lambda_i + \lambda) \boldsymbol{v}_i
= \sum_{i\in I} \mu_i \boldsymbol{v}_i,
\end{align*}
where
$$
\mu_i = \left\{
\begin{array}{ll}
\lambda_i
&\textrm{for } i\in I\setminus(A'\cup B')\\
\lambda_i - \lambda
&\textrm{for } i\in A'\\
\lambda_i + \lambda
&\textrm{for } i\in B'.
\end{array}
\right.
$$
Note that $\mu_i \in \ZZ_{\ge 0}$ for all $i\in I$,
and that $\mu_i = 0$ for at least one index $i \in A' \subseteq I$.
Thus $\boldsymbol{u}$ is expressed as
the integer conical combination of at most $d - 1$ distinct vectors in $\P_c$.

By repeating the above argument,
we can eventually obtain a representation of $\boldsymbol{u}$
as the integer conical combination of at most $d(c)$ distinct vectors in $\P_c$.
\end{proof}

If there is a multiset of vectors in $\P_c$ summing up to $\boldsymbol{a}$,
then
$\boldsymbol{a}$ is an integer conical combination of vectors in $\P_c$,
and hence is included in $\mathrm{int.cone}(\P_c)$,
and hence by Lemma~\ref{lem:cone}
can be represented as an integer conical combination of
$d \le d(c)$ distinct vectors in $\P_c$.
Thus to decide
whether there exists
a multiset of vectors in $\P_c$ whose sum is $\boldsymbol{a}$,
we can enumerate all subsets of $\{1,\ldots,p(c)\}$ of size $d \le d(c)$,
then for each subset $I$, check whether
$\sum_{i\in I} \lambda_i \boldsymbol{v}_i = \boldsymbol{a}$ has a solution
with $\lambda_i \in \ZZ_{\ge 0}$ for all $i\in I$.

The checking step reduces to an integer linear program,
with $O(c)$ linear constraints on at most $d(c) = O(c)$ variables $\lambda_i$,
and with integer coefficients at most $kc$, where $\log kc  = O(N)$.

We now analyze the running time of our algorithm.
Recall that converting the item weights $w_i$ to multiplicities $a_w$ takes $N^{O(1)}$ time.
Finding the $p(c)$ vectors in $\P_c$,
$\boldsymbol{v}_i$ for $1 \le i \le p(c)$,
takes $c^{O(c)}$ time.
The number of subsets of $\{1,\ldots,p(c)\}$ of size $d \le d(c)$ is at most
$(p(c))^{d(c)} = (2^{O(\sqrt{c})})^{O(c)} = 2^{O(c^{3/2})}$.

Solving an integer linear program with $n$ variables and $m$ constraints
and with maximum coefficient $\Delta$
takes $n^{O(n)} m^{O(1)} (\log\Delta)^{O(1)}$ time\footnote{See discussion at
\url{https://cstheory.stackexchange.com/questions/16530}
for more refined estimates.}~\cite{Le83,Ka87}.
In our setting, this is
$(O(c))^{O(c)} (O(c))^{O(1)} (O(N))^{O(1)} = c^{O(c)}N^{O(1)}$.

Thus the total running time is
$N^{O(1)} + c^{O(c)} + 2^{O(c^{3/2})}\cdot c^{O(c)}N^{O(1)}$,
which simplifies to
$2^{O(c^{3/2})}N^{O(1)}$,
since $c^{O(c)} = 2^{O(c\log c)}$.
This completes the proof of Theorem~\ref{thm:fpt}.

\newpage
\section{Dual approximation algorithm for bin packing}

In this section we prove Theorem~\ref{thm:dual}.

Fix any $0 < \epsilon \le 1$.
Separate the $n$ items into two groups:
\emph{small} items of weight at most $\epsilon c$,
and
\emph{large} items of weight greater than $\epsilon c$.

Let $c' = \lfloor 2\epsilon^{-2} \rfloor$.
Note that $c' \cdot \epsilon^2 c / 2 \le c$.
For each large item of weight $w_i > \epsilon c$,
let $w'_i$ be the largest integer such that
$w'_i\cdot \epsilon^2 c / 2 \le w_i$.
Then $(w'_i + 1)\cdot \epsilon^2 c / 2 > w_i$.

The algorithm works as follows.
If the total weight of the $n$ items exceeds $kc$,
then clearly the $n$ items cannot be packed into $k$ bins of capacity $c$,
so return no.
Otherwise,
construct a set of scaled items including
one item of weight $w'_i$
for each large item of weight $w_i$.
Use the fixed-parameter algorithm in Theorem~\ref{thm:fpt}
to decide
whether the set of at most $n$ scaled items $w'_i$ can be packed into $k$ scaled
bins of capacity $c'$,
then return the answer accordingly.

\begin{lemma}
If the scaled items can be packed into $k$ scaled bins
of capacity $c'$,
then the $n$ items can be packed into
$k$ enlarged bins of capacity $(1+\epsilon)c$.
\end{lemma}

\begin{proof}
For each large item of weight $w_i > \epsilon c$,
the corresponding scaled item has weight
$$
w'_i > \frac{w_i}{\epsilon^2 c / 2} - 1
> \frac{\epsilon c}{\epsilon^2 c / 2} - 1 = 2\epsilon^{-1} - 1 \ge \epsilon^{-1}.
$$
Thus each scaled bin of capacity $c' = \lfloor 2\epsilon^{-2} \rfloor$
can hold at most
$\lfloor 2\epsilon^{-2} \rfloor / \epsilon^{-1} \le 2\epsilon^{-1}$
scaled items.

For the scaled items in each scaled bin,
the sum of their weights $w'_i$ is at most $c'$,
and hence the sum of $w'_i\cdot \epsilon^2 c / 2$ is at most
$c'\cdot \epsilon^2 c / 2 \le c$.
Recall that $w'_i\cdot \epsilon^2 c / 2 \le w_i < (w'_i + 1)\cdot \epsilon^2 c / 2$,
and hence $w_i - w'_i\cdot \epsilon^2 c / 2 < \epsilon^2 c / 2$.
Corresponding to the at most $2\epsilon^{-1}$ scaled items in each scaled bin,
the total weight of the large items can exceed $c$
by at most $2\epsilon^{-1} \cdot \epsilon^2 c / 2= \epsilon c$,
and hence is at most $(1 + \epsilon)c$.
Thus we can pack all large items
into $k$ enlarged bins of capacity $(1 + \epsilon)c$.

Since the total weight of both large and small items is at most $kc$,
at least one of the $k$ enlarged bins is filled to at most $c$ of its capacity
$(1+\epsilon)c$,
and hence can always accommodate one more small item of weight at most
$\epsilon c$.
Thus we can pack all small items into the $k$ enlarged bins
after the large items.
\end{proof}

\begin{lemma}
If the scaled items cannot be packed into $k$ scaled bins
of capacity $c'$,
then the $n$ items cannot be packed into
$k$ bins of capacity $c$.
\end{lemma}

\begin{proof}
We prove the contrapositive.
Suppose that the $n$ items of total weight at most $kc$ can be packed into $k$ bins
of capacity $c$.
Then by discarding the small items and rounding down the weights $w_i$
of the large items to integer multiples of $\epsilon^2 c / 2$,
the rounded items of weights $w'_i \cdot \epsilon^2 c / 2$
can be packed into $k$ bins of capacity
$c = 2\epsilon^{-2} \cdot \epsilon^2 c / 2$,
and hence the scaled items of integer weights $w'_i$
can be packed into $k$ scaled bins of capacity
$c' = \lfloor 2\epsilon^{-2} \rfloor$.
\end{proof}

With at most $n$ scaled items
and capacity $c' = O(\epsilon^{-2})$ for the $k$ scaled bins,
it follows by Theorem~\ref{thm:fpt} that the running time
of our dual approximation algorithm is
$2^{O(\epsilon^{-3})}N^{O(1)}$,
where $N$ is the size of the problem instance encoded in binary.
This completes the proof of Theorem~\ref{thm:dual}.

\section{Parameterized complexity of decomposing trees}

In this section we prove Theorem~\ref{thm:tree}.
The problem \kcd\ is clearly in NP\@.
In the following, we prove that when both $k$ and $c$ are part of the input,
\kcd\ in trees is not only NP-hard but also W[1]-hard with parameter $k$,
and is FPT with parameter $c$.

\subsection{NP-hardness and W[1]-hardness with parameter $k$}

We give a polynomial reduction from
\ubp\ with $k$ bins to \kcd\ in trees.
Recall that \ubp\ is W[1]-hard with parameter $k$~\cite{JKMS13}.

Given $n$ items of weights $w_i$, $1 \le i \le n$,
and $k$ bins of capacity $c$,
we construct a tree $T$ with $n(k-1)c + \sum_{i=1}^n w_i$ edges as follows.
For each item of weight $w_i$, make a star with $(k-1)c + w_i$ edges.
Then select an arbitrary leaf from each star,
and merge the $n$ leaves selected from the $n$ stars into one center vertex,
denoted by $v$.

We claim that
the $n$ items can be packed into $k$ bins of capacity $c$
if and only if
the tree $T$ can be decomposed into $k$ subgraphs
with component size at most $c$.

We first prove the direct implication of the claim.
Suppose that the $n$ items can be packed into $k$ bins of capacity $c$.
We color the edges of the tree $T$ with $k$ colors as follows.
First color the $n$ edges incident to the center vertex $v$ with $k$ colors
corresponding to the $n$ items in the $k$ bins.
Then for each star with $(k-1)c + w_i$ edges,
color $w_i - 1$ additional edges with the same color
as the edge incident to the center vertex $v$,
and color the remaining $(k-1)c$ edges with the other $k-1$ colors,
$c$ edges of each color.
Then the edges of each color induce a forest with component size at most $c$.

We next prove the reverse implication of the claim.
Suppose that the tree $T$ can be decomposed into $k$ subgraphs with
component size at most $c$.
Color the edges of $T$ with $k$ colors corresponding to the $k$ subgraphs.
Then in the star with $(k-1)c + w_i$ edges
corresponding to the item of weight $w_i$,
the edge incident to the center vertex $v$
must have the same color as at least $w_i - 1$ other edges,
because the other $k-1$ colors can accommodate at most $(k-1)c$ edges.
Note that for each of the $k$ colors,
the edges of this color that are incident to $v$,
as well as the other edges of the same color in their respective stars,
are all connected,
and form one component of size at most $c$.
Put the $n$ items into $k$ bins corresponding to the colors of the $n$ edges
incident to $v$.
Then the items in each bin must have total weight at most $c$ too.

\subsection{FPT with parameter $c$}

We present a fixed-parameter algorithm for \kcd\ in trees with parameter $c$.
Let $T$ be a tree with $n$ vertices and $n - 1$ edges,
rooted at an arbitrary vertex of degree $1$.

For each edge $e$ between a vertex $v$ and its parent in $T$,
denote by $T_e$ the subtree of $T$ rooted at $v$ plus the edge $e$,
and denote by $s(e)$ the minimum $s \ge 1$ such that
$T_e$ admits an edge partition into $k$ subgraphs
with component size at most $c$,
under the constraint that $e$ is in a component of size at most $s$.
If no such $s$ exists, let $s(e) = c + 1$.
Then $T$ admits an edge partition into $k$ subgraphs
with component size at most $c$ if and only if
$s(e) \in [1,c]$ for all edges $e$.

We can compute $s(e)$ by dynamic programming,
following a post-order traversal of the rooted tree $T$.
For each edge $e$ incident to a leaf,
we clearly have $s(e) = 1$.
For each edge $e$ between an internal node $v$ and its parent,
and for a candidate size $s \ge 1$,
we can check whether $s(e) \le s$ as follows.

Let $m_v$ be the number of edges incident to $v$.
Construct $m_v$ items corresponding to the $m_v$ edges.
Each edge $f$ between $v$ and a child corresponds to an item of weight $s(f)$,
which has been computed following the post-order traversal.
The edge $e$ between $v$ and its parent corresponds to an item of weight $c-s + 1$,
where the surplus of $c-s$ in addition to $1$ for the edge $e$
accounts for the difference between the desired component size $s$ for $e$
and the upper bound $c$.
Then $T_e$ admits an edge partition into $k$ subgraphs
with component size at most $c$,
where $e$ is in a component of size at most $s$,
if and only if
the $m_v$ items can be packed into $k$ bins of capacity $c$.

By Theorem~\ref{thm:fpt},
there is a fixed-parameter algorithm for \bp\ that decides
whether the $m_v$ items thus constructed
can be packed into $k$ bins of capacity $c$
in $2^{O(c^{3/2})}N_v^{O(1)}$ time,
where $N_v = O(m_v(1 + \log kc))$ is the encoding size of the \bp\ instance with $m_v$ items.
By one invocation of this fixed-parameter algorithm with $s = c$,
we either find that $s(e) > c$ and abort the traversal of $T$,
or find that $s(e) \le c$.
In the latter case,
we can then determine the exact value of $s(e)$ in $[1,c]$
by a binary search,
using the fixed-parameter algorithm as a decision procedure.

The total running time of the algorithm is
$$
n^{O(1)} + \sum_v \big( O(1 + \log c) \cdot 2^{O(c^{3/2})}N_v^{O(1)} \big)
= 2^{O(c^{3/2})}n^{O(1)}.
$$
Thus \kcd\ in trees is FPT with parameter $c$.
This completes the proof of Theorem~\ref{thm:tree}.

\section{Approximation algorithms for decomposing trees}

In this section we prove Theorem~\ref{thm:mcd} and Theorem~\ref{thm:msd}.
Let $T$ be a tree with $n$ vertices, $n-1$ edges, and maximum degree $\Delta < n$,
rooted at an arbitrary vertex of degree $1$.

\subsection{Approximating \mcd}

We first present a linear-time approximation algorithm for \mcd.

Perform a pre-order traversal of the rooted tree $T$.
The root is incident to only one edge, which is colored arbitrarily.
At each vertex other than the root,
color the at most $\Delta-1$ edges to its children with
at most $\lceil \frac{\Delta - 1}c \rceil$ colors
(at most $c$ edges of each color),
all different from the color of the edge from its parent.

It is straightforward to verify that
the edges of each color induce a forest with component size at most $c$,
and moreover each component is a star with at most $c$ edges.

Let $k' = \lceil \frac{\Delta}c \rceil$
and $k'' = \lceil \frac{\Delta - 1}c \rceil + 1$.
The minimum number of subgraphs with component size at most $c$
into which $T$ can be decomposed is at least $k'$.
On the other hand, our algorithm colors the edges of $T$ with $k''$ colors,
where $k'' \le k' + 1$.
This completes the proof of Theorem~\ref{thm:mcd}.

\subsection{Approximating \msd}

We next present a polynomial-time approximation scheme for \msd.

We first note that
our approximation algorithm for \mcd\ can be slightly modified to yield
a $2$-approximation for \msd\ as follows.
During the pre-order traversal of the rooted tree $T$, at each vertex other than the root,
color the at most $\Delta-1$ edges to its children with at most $k-1$ colors
(at most $\lceil \frac{\Delta - 1}{k-1} \rceil$ edges of each color),
all different from the color of the edge from its parent.
Let $c' = \lceil \frac{\Delta}k \rceil$
and $c'' = \lceil \frac{\Delta - 1}{k-1} \rceil$.
In any edge partition of $T$ into $k$ subgraphs,
there is a subgraph with component size at least $c'$.
On the other hand, our algorithm colors the edges of $T$ with $k$ colors
such that the edges of each color induce a forest with component size at most
$c''$, where
$$
c'' = \left\lceil \frac{\Delta - 1}{k-1} \right\rceil
= \left\lceil \frac{k}{k - 1} \cdot \frac{\Delta - 1}k \right\rceil
\le \left\lceil \frac{k}{k - 1} \right\rceil \cdot \left\lceil \frac{\Delta - 1}k \right\rceil
\le 2\left\lceil \frac{\Delta}k \right\rceil
= 2c'.
$$

Now fix any $0 < \epsilon \le 1$.
To obtain a $(1+\epsilon)$-approximation of the optimal component size $c^*$,
it suffices to perform a binary search in the integer range $(c' - 1, c'']$,
using a subroutine that decides
either $c^* > c$ or $c^* \le (1+\epsilon)c$ for any candidate component size $c$.
When the search range is reduced to $(l, h]$ with $h - l = 1$,
we must have $c^* \ge h$,
and the subroutine with $c = h$ decomposes the tree into $k$ subgraphs with component size
at most $(1+\epsilon)\,h$.

\bigskip
The subroutine on a candidate component size $c$
is similar to our fixed-parameter algorithm with parameter $c$
in Theorem~\ref{thm:tree}.
As before, for each edge $e$ between a vertex $v$ and its parent in the rooted tree $T$,
denote by $T_e$ the subtree of $T$ rooted at $v$ plus the edge $e$,
and denote by $s(e)$ the minimum $s \ge 1$ such that
$T_e$ admits an edge partition into $k$ subgraphs
with component size at most $c$,
under the constraint that $e$ is in a component of size at most $s$.
If no such $s$ exists, let $s(e) = c + 1$.

By definition, if $s(e) \le c$,
then $s(e)$ is the unique integer $s \in [1,c]$
that satisfies the following two conditions:
\begin{enumerate}\setlength\itemsep{0pt}

\item[A.]
$T_e$ admits an edge partition into $k$ subgraphs
with component size at most $c$,
where $e$ is in a component of size at most $s$.

\item[B.]
$T_e$ does not admit any edge partition into $k$ subgraphs
with component size at most $c$,
where $e$ is in a component of size at most $s - 1$.

\end{enumerate}

Instead of computing $s(e)$ directly,
the subroutine tries to find a \emph{feasible} value $s\in[1,c]$
that satisfies the following two conditions,
and records it in $\tilde s(e)$:
\begin{enumerate}\setlength\itemsep{0pt}

\item[C.]
$T_e$ admits an edge partition into $k$ subgraphs
with component size at most $(1+\epsilon)c$,
where $e$ is in a component of size at most $s$.

\item[B.]
$T_e$ does not admit any edge partition into $k$ subgraphs
with component size at most $c$,
where $e$ is in a component of size at most $s - 1$.

\end{enumerate}

The subroutine computes $\tilde s(e)$ by dynamic programming,
following a post-order traversal of the rooted tree $T$.
For each edge $e$ incident to a leaf,
set $\tilde s(e) \gets 1$.
For each edge $e$ between an internal node $v$ and its parent,
check whether there exists a feasible value $s \in [1,c]$
for $\tilde s(e)$
that satisfies conditions C and B, as follows.

Let $m_v$ be the number of edges incident to $v$.
Construct $m_v$ items corresponding to the $m_v$ edges.
Each edge $f$ between $v$ and a child corresponds to an item of weight $\tilde s(f)$.
The edge $e$ between $v$ and its parent corresponds to an item of weight $c-s + 1$.

By Theorem~\ref{thm:dual},
there is a dual approximation algorithm for \bp\ that,
given the $m_v$ items and $k$ bins of capacity $c$,
either
decides that the $m_v$ items can be packed into $k$ enlarged bins
of capacity $(1+\epsilon)c$ and returns yes,
or
decides that the $m_v$ items cannot be packed into $k$ bins of capacity $c$
and returns no.

First call this dual approximation algorithm with $s = c$.
If its answer is no,
then abort the traversal of $T$, and declare that $\hl{c^* > c}$.
Otherwise,
find a feasible value $s \in [1,c]$ for $\tilde s(e)$
that satisfies both conditions C and B by a binary search,
using the dual approximation algorithm as a decision procedure.

If the traversal of $T$ finishes successfully,
declare that $\hl{c^* \le (1+\epsilon)c}$.

\bigskip
The correctness of the subroutine can be verified by induction
following the post-order traversal.
In particular,
for all edges $e$ such that a feasible value $s \in [1,c]$ is found for $\tilde s(e)$,
we have $\tilde s(e) \le s(e)$ because
they satisfy the same condition B,
while condition C is a relaxation of condition A.
Also observe that the answers of the dual approximation algorithm
have the following two implications:
\begin{itemize}\setlength\itemsep{0pt}

\item
If the $m_v$ items can be packed into $k$ enlarged bins of capacity
$(1+\epsilon)c$, then
$T_e$ admits an edge partition into $k$ subgraphs
with component size at most $(1+\epsilon)c$,
where $e$ is in a component of size at most $s$.

\item
If the $m_v$ items cannot be packed into $k$ bins of capacity $c$, then
$T_e$ does not admit any edge partition into $k$ subgraphs
with component size at most $c$,
where $e$ is in a component of size at most $s$.

\end{itemize}

By Theorem~\ref{thm:dual},
the running time of the dual approximation algorithm
on the $m_v$ items is
$2^{O(\epsilon^{-3})}N_v^{O(1)}$,
where $N_v = O(m_v(1 + \log kc))$ is the encoding size of the problem instance.
The total running time of the algorithm,
including the nested binary searches on $c^* \in (c'-1,c'']$ and on $s \in [1,c]$,
is
$$
n^{O(1)} +
O(\log(\Delta/k))\Big( n^{O(1)} + \sum_v \big( O(1 + \log c) \cdot 2^{O(\epsilon^{-3})}N_v^{O(1)} \big) \Big)
= 2^{O(\epsilon^{-3})}n^{O(1)}.
$$
Thus we have a polynomial-time approximation scheme for \msd.
This completes the proof of Theorem~\ref{thm:msd}.

\section{Concluding remark}

In our study of graph decomposition problems in this paper,
the focus is on bipartite graphs and trees.
It may be interesting to investigate the complexities of these problems
on other important graph classes too.

\end{document}